\documentclass{amsart}
\usepackage{color}
\usepackage{graphicx,epsf}
\usepackage{amsmath,amscd,amssymb}
\usepackage{enumerate}
\usepackage{cancel}
\newtheorem{theorem}{Theorem}[section]
\newtheorem{lemma}[theorem]{Lemma}

\newtheorem{proposition}[theorem]{Proposition}
\theoremstyle{definition}
\newtheorem{definition}[theorem]{Definition}
\newtheorem{example}[theorem]{Example}

\theoremstyle{remark}
\newtheorem{remark}[theorem]{Remark}

\numberwithin{equation}{section}

\newcommand{\Real}{{\mathbb R}}
\newcommand{\Z}{\mathbb Z}

\newcommand{\eps}{\varepsilon}

\newcommand{\x}{\mathbf{x}}

\newcommand {\hide}[1]{}

\begin{document}
\title[Topological lower bounds for arithmetic networks]
{Topological lower bounds for arithmetic networks}

\author{Andrei Gabrielov}
\address{Department of Mathematics,
Purdue University, West Lafayette, IN 47907, USA}
\email{agabriel@math.purdue.edu}
\author{Nicolai Vorobjov}
\address{
Department of Computer Science, University of Bath, Bath
BA2 7AY, England, UK}
\email{nnv@cs.bath.ac.uk}

\begin{abstract}
We prove that the depth of any arithmetic network for deciding membership in a semialgebraic
set $\Sigma \subset \Real^n$ is bounded from below by
$$c_1 \sqrt{ \frac{\log ({\rm b}(\Sigma))}{n}} -c_2 \log n,$$
where ${\rm b}(\Sigma)$ is the sum of the Betti numbers of $\Sigma$ with respect to ``ordinary'' (singular) homology,
and $c_1,\ c_2$ are some (absolute) positive constants.
This result complements the similar lower bound by Monta\~na, Morais and Pardo \cite{Mon2} for
{\em locally closed} semialgebraic sets in terms of the sum of {\em Borel-Moore} Betti numbers.

We also prove that if $\rho:\> \Real^n \to \Real^{n-r}$ is the projection map, for some $r=0, \ldots , n$, then
the depth of any arithmetic network deciding membership in $\Sigma$ is bounded by
$$\frac{c_1\sqrt{\log ({\rm b}(\rho(\Sigma)))}}{n} - c_2 \log n$$
for some positive constants $c_1,\ c_2$.
\end{abstract}
\maketitle

\section{Introduction}

An {\em arithmetic network} is a computational model aimed to capture the idea of a parallel computation
in its simplest form.
It was first proposed by J. von zur Gathen \cite{Gathen}.

We will be dealing with a special class of networks, called {\em decision} arithmetic networks.
The next definition follows \cite{Gathen} and \cite{Mon1}.
It combines an arithmetic circuit (straight-line program) and a Boolean circuit by means of
two special types of gates: sign and selection.

\begin{definition}\label{def:network}
A decision arithmetic network $\mathcal N$ over $\Real$ is an acyclic directed graph with vertices
(called ``gates'') of the following four types, classified by the indegree.
\begin{enumerate}
\item
Gates of indegree 0 are either {\em input gates} or {\em constant gates}.
The latter are labelled by real numbers.
Input and constant gates have the {\em output type} of a real number.
\item
Gates of indegree 1 are either the (unique) {\em output gate}, or {\em sign gates} labelled by one of
the signs $<,=,>$, or {\em Boolean gates} labelled by $\neg$.
Output gate has the Boolean input type (i.e., {\bf true} or {\bf false}) and the outdegree 0.
Sign gates have input type of a real number and output type Boolean.
Boolean gates labelled by $\neg$ have input type Boolean and output type Boolean.
\item
Gates of indegree 2 are either {\em arithmetic gates}, labelled by one of arithmetic operations
$+, -, \times$, or {\em Boolean gates}, labelled by one of Boolean operations $\lor, \land$.
Arithmetic gates have input and output types both of a real number, Boolean gates have input and output
types both Boolean.
\item
Gates of indegree 3 are {\em selection gates}.
In each such gate two of its inputs have the type of a real number while the third input has the Boolean type which comes
from a Boolean or a sign gate.
Output type is of a real number.
(See Fig.~4, on which $g,\ h$ and $f$ are real numbers, while $A$ is a Boolean value.)
\end{enumerate}

The outdegrees for all gates, except output ones, can be arbitrary positive numbers.

The operational semantics of $\mathcal N$ is clear, except maybe the functioning of the sign and selection gates.
Let $v$ be a sign gate labelled by a sign $\sigma \in \{ <,=,> \}$.
Then for an input $f \in \Real$ of $v$, the output of $v$ is {\bf true} if $f\ \sigma\ 0$ and {\bf false} otherwise.
Let $w$ be a selection gate, and $(g,h, {\bf b}) \in \Real \times \Real \times \{ {\bf true}, {\bf false} \}$
be its input.
Then the output of $w$ is $g$ if ${\bf b}={\bf true}$ and $h$ otherwise.

Let the number of input gates be $n$, and let $(x_1, \ldots ,x_n) \in \Real^n$ be a particular input.
We say that $\mathcal N$ {\em accepts} $(x_1, \ldots, x_n)$ if the Boolean value at the output gate is
{\bf true}.
\end{definition}

We interpret the {\em complexity (parallel time)} of $\mathcal N$ as its {\em depth}, defined as follows.

\begin{definition}
The {\em size} $s(\mathcal N)$ of an arithmetic network $\mathcal N$ is the number of all its gates.
The {\em depth} $d(\mathcal N)$ of $\mathcal N$ is the length of the longest (directed) path from
some input gate to the output gate.
Each gate $v$ of $\mathcal N$ also has the {\em depth}, which is the length of the longest path from
some input gate to $v$.
\end{definition}

It is clear that $s(\mathcal N) \le 3^{d(\mathcal N)+1}$.
We will be mostly interested in lower bounds for $d(\mathcal N)$.

\begin{remark}
In \cite{BC, Mon2} one can find an equivalent definition of an arithmetic network, which combines
a straight-line program with ``sign gates''.
Such ``sign gate'' takes an input of the type of a real number and outputs 1 if the input is positive
and 0 otherwise.
Boolean, sign, and selection gates from Definition~\ref{def:network} can be modelled by arithmetic and ``sign gates''
in the alternative model in a straightforward way.
For our purposes this encoding places an additional layer on the semantics and does not appear to
give any advantages.
Two models can simulate each other with respect to both size and depth within a constant factor \cite{BC}.
\end{remark}

In Proposition~\ref{pr:semialg-to-network} below, we will associate with each arithmetic network $\mathcal N$,
having $n$ input gates, a formula $B(\mathcal N)$ of the first order theory of the reals, such that $\mathcal N$
accepts exactly all elements in the semi-algebraic set $\Sigma \subset \Real^n$ defined by $B(\mathcal N)$.
We say that $\mathcal N$ {\em tests membership} in $\Sigma$.
For brevity, in what follows, we will refer to $B(\mathcal N)$ (and to any other first order formula) as to {\em Boolean} formula.

The first topological lower bound for depths of arithmetic networks appeared in \cite{Mon1}.
\begin{proposition}[\cite{Mon1}]\label{pr:Mon1}
Let $\mathcal N$ be an arithmetic network testing membership in a semi-algebraic set
$\Sigma \subset \Real^n$.
Then
$$
d({\mathcal N}) = \Omega \left( \sqrt{ \frac{\log ({\rm b}_0(\Sigma))}{n}} \right),
$$
where ${\rm b}_0(\Sigma)$ is the number of connected components of $\Sigma$.
\end{proposition}

This result was then partly expanded in \cite{Mon2} as follows.

\begin{proposition}[\cite{Mon2}]\label{pr:Mon2}
Let $\mathcal N$ be an arithmetic network testing membership in a locally closed semi-algebraic set
$\Sigma \subset \Real^n$.
Then
\begin{equation}\label{eq:Mon2}
d({\mathcal N}) = \Omega \left( \sqrt{ \frac{\log ({\rm b}^{BM}(\Sigma))}{n}} \right),
\end{equation}
where ${\rm b}^{BM}(\Sigma)$ is the sum of Borel-Moore Betti numbers of $\Sigma$.
\end{proposition}

Note that for {\em compact} sets, Borel-Moore homologies coincide with ``ordinary'' singular homologies,
in particular, Betti numbers coincide (see, e.g., \cite{Mon2}).
For locally closed sets, Borel-Moore and singular Betti numbers are incomparable.
The proofs of bounds in Propositions~\ref{pr:Mon1} and \ref{pr:Mon2} depend heavily on subadditivity
property of Borel-Moore Betti numbers which is not generally valid for singular Betti numbers.

Lower bounds for the membership in locally closed semi-algebraic sets with respect to the algebraic computation tree
(a sequential model), in terms of Borel-Moore Betti numbers, were obtained in \cite{Yao}.

In this paper we present two main results.
Theorem~\ref{th:general} states a lower bound similar to (\ref{eq:Mon2}) for {\em arbitrary} (not necessarily
locally closed) semi-algebraic sets in terms of singular homology.
Theorem~\ref{th:proj} suggests a lower bound of a new type, in terms of singular homology Betti numbers
of a {\em projection} of a semi-algebraic set to a coordinate subspace.
Note that the topology of the image under a projection may be much more complex than the
topology of the set being projected.
We are not aware of previous lower bounds of this sort.
This generalization comes at a price of lowering the bound (as compared to Theorem~\ref{th:general}), namely
the term $n$, rather than $\sqrt{n}$, appears in the denominator.

Note that lower bounds for membership in arbitrary semialgebraic sets with respect to algebraic computation
trees in terms of singular homology were obtained in \cite{GV14}.

\section{Further properties of arithmetic networks}

In this section we establish some properties of arithmetic networks, which we need further in the paper.

\subsection{Associating Boolean formula to a network}\label{sec:denotational}
We will now associate with each arithmetic network $\mathcal N$ a Boolean formula $B(\mathcal N)$.
This semantics follows \cite{Gathen}, and is simplified and adjusted for our modification of arithmetic networks.

\begin{definition}\label{def:piecewise}
Consider a finite partition of a semi-algebraic set $S$ into semi-algebraic sets $S_1, \ldots, S_k$.
A {\em partial piecewise polynomial function} (or just {\em partial piecewise polynomial}) {\em with respect to the partition},
$f:\> S \to \Real$, coincides
on each $S_i$ with the restriction of some polynomial function $f_i:\> \Real^n \to \Real$ to $S_i$.
If $S= \Real^n$ we drop the expression ``partial'' in this definition.
Arithmetic operations with partial piecewise polynomials having the same domain, and predicates,
$f\ \sigma\ 0,\ f\ \cancel{\sigma}\ 0$, where $\sigma \in \{ <,=,> \}$
are defined in a usual way, as operations and predicates on functions.
\end{definition}

\begin{definition}\label{def:pieceboolean}
Let $f$ be a partial piecewise polynomial function with the partition $S_1, \ldots, S_k$ such that for each $i$
the set $S_i$ is defined by a Boolean formula $B_i$ with atomic subformulae of the kind $g\ \sigma\ 0$,
where $\sigma \in \{ <,=,> \}$, and $g$ is a polynomial.
We use the same notation $f\ \sigma\ 0$ also for the Boolean formula
$$
((f_1\ \sigma\ 0) \land B_1) \lor \cdots \lor ((f_k\ \sigma\ 0) \land B_k),
$$
which describes the predicate $f\ \sigma\ 0$.
We also say that the partial piecewise polynomial $f$ is described by the list $f_1, \ldots, f_k; B_1. \ldots,B_k$.
\end{definition}

We now describe the Boolean formula $B(\mathcal N)$ associated to a network $\mathcal N$.
We associate a polynomial (namely, a variable) to each input gate
and a Boolean formula to each sign gate, each Boolean gate and to the output gate.
With each selection gate and each arithmetic gate we associate a piecewise polynomial.

We perform these associations by induction on the depth of the gate as follows.

Input gates are assigned their input variables, while constant gates -- their constants.
This completes the base of the induction.

Now we perform the induction step.

\begin{itemize}
\item
If $v$ is an arithmetic gate labelled by $\ast \in \{+,-, \times \}$, then it has two parents.
Parents may be either arithmetic or selection or input or constant gates in any combination.
In any case, the parents have associated piecewise polynomials, say $f$ and $g$.
Associate with $v$ the piecewise polynomial $f \ast g$.
\item
If $v$ is a Boolean gate labelled by $\dag \in \{ \lor, \land \}$, then it has two parents, either sign or
Boolean gates, in any combination, with associated Boolean formulae, say $A$ and $B$.
Associate with $v$ the Boolean formula $A\ \dag\ B$.
\item
If $v$ is a Boolean gate labelled by $\neg$, then it has one parent, either a sign or a Boolean gate, with
associated Boolean formula, say $B$.
Associate with $v$ the Boolean formula $\neg B$.
\item
If $v$ is a sign gate labelled by $\sigma \in \{<,=,> \}$, then it has either one arithmetic or one selection gate
parent with associated piecewise polynomial, say $f$.
Associate with $v$ the Boolean formula $f\ \sigma\ 0$ (see Definition~\ref{def:pieceboolean}).
\item
If $v$ is a selection gate, then it has three parents.
Two of them are either input, or constant, or arithmetic or selection gates
(in any combination) with associated piecewise polynomial functions, say $f$ and $g$.
The third is a Boolean parent gate with associated Boolean formula, say $B$.
Let $\mathcal F$ (respectively, $\mathcal G$) be the partition of $\Real^n$ corresponding to $f$
(respectively, to $g$).
Each partition ${\mathcal F},\ {\mathcal G}$ is represented by a list of Boolean formulae, each representing
an element of the partition.
Consider the partition $\mathcal H$ of $\Real^n$ whose elements are all intersections of the kind $U \cap V$,
where $U \in {\mathcal F}$ and $V=\{ \x \in \Real^n|\> B \}$,  or
$U \in {\mathcal G}$ and $V= \{ \x \in \Real^n|\> \neg B\}$.
Then associate with $v$ the piecewise polynomial $h$, having the partition $\mathcal H$, and coinciding with
$f$ on $\{ \x \in \Real^n|\> B \}$ and with $g$ on $\{ \x \in \Real^n|\> \neg B \}$.
\item
If $v$ is the output gate, then it has one Boolean gate parent with an associated Boolean formula, say $B$.
Associate with $v$ the same Boolean formula $B$.
\end{itemize}

Associate with $\mathcal N$ the Boolean formula associated with the output gate, and denote this formula by
$B({\mathcal N})$.

The following statement is proved in \cite{Gathen}.
Here we give a proof which uses some concepts we will need further on.

\begin{proposition}\label{pr:semialg-to-network}
A set $\Sigma \subset \Real^n$ is semi-algebraic if and only if there is an arithmetic network $\mathcal N$ with $n$
input gates, accepting exactly all inputs in $\Sigma$.
\end{proposition}

\begin{proof}
Let $\Sigma$ be a semi-algebraic set represented by a disjunctive normal form with atomic polynomial equations
and strict inequalities.
One can construct an arithmetic network for $\Sigma$ as follows.
Compute ``in parallel'' each atomic polynomial using straight-line programs (arithmetic circuits).
Attach to the output gate of each program a sign gate labelled by the sign of the corresponding atomic
formula.
Then evaluate the resulting Boolean disjunctive normal form using a Boolean circuit with inputs from
all sign gates.

To prove the converse statement, notice that every arithmetic network $\mathcal N$ accepts exactly all inputs
in the semi-algebraic set in $\Real^n$ defined by the Boolean formula $B(\mathcal N)$.
\end{proof}

\subsection{Elimination of negations}\label{sec:elim}

\begin{definition}
Two arithmetic networks $\mathcal N$ and ${\mathcal N}'$ with $n$ inputs each are {\em equivalent} if
the sets of all accepted inputs for $\mathcal N$ and ${\mathcal N}'$ coincide.
\end{definition}

Clearly two equivalent networks test membership in the same semi-algebraic set.

\begin{lemma}\label{le:negation}
For every arithmetic network $\mathcal N$ there is an equivalent arithmetic network ${\mathcal N}'$ having
no Boolean gates labelled by $\neg$, and such that $d({\mathcal N}') \le d({\mathcal N})$.
\end{lemma}

\begin{proof}
The idea of the proof is to push computing of negations to sign gates, where negations of associated Boolean formulae of the kind
$f\ \sigma\ 0$ can be replaced by formulae $f\ \cancel{\sigma}\ 0$, thus avoiding an explicit use of negation gates
(cf. Example~\ref{ex:neg}).

We start with the inductive construction of a directed graph ${\mathcal M}_{\ell +1}$, obtained  by attaching new sign or Boolean
gates of depth at most $\ell +1$ to graph ${\mathcal M}_{\ell}$, constructed by the inductive hypothesis.
Each graph ${\mathcal M}_{\ell}$ differs from an arithmetic network in that it may contain some {\em hanging vertices}
(i.e., gates having no outgoing edges, but different from the output gate).
For such graphs we can define semi-algebraic sets of accepting inputs, associated Boolean formulae, and the equivalence relation
exactly as for the networks.
Each graph ${\mathcal M}_{\ell}$ will be equivalent to ${\mathcal N}$.
We will obtain the network ${\mathcal N}'$ from the last ${\mathcal M}_{\ell}$ in the induction, by removing the irrelevant gates,
including all hanging vertices.

For the base of the induction take ${\mathcal M}_0= \mathcal N$.
The inductive hypothesis assumes the following.
\begin{enumerate}
\item
Let $v$ be a Boolean gate in ${\mathcal M}_\ell$ labelled by $\neg$, having depth $\ell$ and
the associated Boolean formula $F$.
Then $v$ is a hanging vertex, and there exists a sign or a Boolean gate $w$ in ${\mathcal M}_{\ell}$,
having a depth at most $\ell$, with an associated Boolean formula which is equivalent to $F$,
and such that $v$ is not labelled by $\neg$.
\item
Let $v$ be a sign or a Boolean gate in ${\mathcal M}_\ell$, having depth $\ell$ and the associated Boolean formula $F$.
Then there exists a sign or a Boolean gate $u$ in ${\mathcal M}_{\ell}$, having the depth at most $\ell +1$,
with an associated Boolean formula which is equivalent to $\neg F$, and such that $u$ is not labelled by $\neg$.
\item
${\mathcal M}_\ell$ is equivalent to $\mathcal N$.
\end{enumerate}
Observe that since there are no sign or Boolean gates with zero depth, the base of induction, for ${\mathcal M}_0=\mathcal N$,
is trivially true.
We now describe the inductive step by constructing ${\mathcal M}_{\ell +1}$.
Let $v$ be a sign or a Boolean gate in ${\mathcal M}_{\ell}$, with $d(v)=\ell +1$.

Suppose first that $v$ is a sign gate.
Then the Boolean formula $F$, associated with $v$, is of the kind $f\ \sigma\ 0$, where $f$ is a piecewise polynomial
and  $\sigma \in \{ <,=,> \}$.
Let, for definiteness, $\sigma$ be $<$ (cases of $=$ and $>$ are considered analogously).
Since, $v$ is not a Boolean gate labelled by $\neg$, the property (1) for $v$ is trivially satisfied.
Suppose that there are no sign gates $v_1$ and $v_2$ in ${\mathcal M}_{\ell}$, of depths at most $\ell +1$, with which
the formulae $f = 0$ and $f>0$ respectively, are associated.
Let $w$ be the parent of $v$ in ${\mathcal M}_{\ell}$ (which is necessarily either an input or a constant or an arithmetic
or a selection gate).
Note that $d(w)= \ell$.
Attach to $w$, by means of an outgoing edges, two new sign gates, $v_1,\ v_2$ to which associate formulae
$f = 0$ and $f>0$ respectively.
Add to a new Boolean gate $w$, labelled by $\lor$, having $v_1,\ v_2$ as parents.
Clearly, $d(w)=\ell +2$, and the condition (2) is satisfied.
Since $w$ is a hanging vertex, the resulting graph is equivalent to ${\mathcal M}_{\ell}$.

If $v$ is a Boolean gate labelled by $\lor$ and $F=A \lor B$, where $A$ and $B$ are Boolean sub-formulae,
then $A$ and $B$ are associated with sign or Boolean gates, say $a$ and $b$,  with depth less than $\ell +1$.
The condition (1) for $v$ is trivially satisfied.

If there is a sign or a Boolean gate, not labelled by $\neg$, with depth at most $\ell +2$, with associated
formula equivalent to $\neg F$, then condition (2) is satisfied.
Suppose otherwise.
By the inductive hypothesis, there are gates $a'$ and $b'$ of depth at most $\ell +1$,
such that the Boolean formula $A'$ (respectively, $B'$) associated with $a'$ (respectively, $b'$)
is equivalent to $\neg A$ (respectively, $\neg B$).
Create a new Boolean gate $v'$, labelled by $\land$, making it a common child of $a'$ and $b'$.
Associate Boolean formula $F'= A' \land B'$ with $v'$.
Observe that $d(v') \le \ell +2$.
Thus, the condition (2) for $v$ is satisfied.
Since $v'$ is a hanging vertex, the resulting graph is equivalent to ${\mathcal M}_{\ell}$.

The similar construction is applied in the case when $F=A \land B$.

Let $v$ be a Boolean gate labelled by $\neg$ with the associated formula $F= \neg A$.
Then formula $A$ is associated with a sign or a Boolean gate $a$ having the depth at most $\ell$.
By the inductive hypothesis, there is a sign or a Boolean gate $a'$, not labelled by $\neg$,
with the depth at most $\ell +1$ and the associated Boolean formula $A'$ such that $A'$ is equivalent to $\neg A$.
Detach the outgoing edges from $v$ and attach them as outgoing edges to $a'$, replacing in Boolean formulae,
associated with the descendants, the subformula $F$ by $A'$.
Clearly, conditions (1) and (2) for $v$ are satisfied, and the resulting graph is equivalent to ${\mathcal M}_{\ell}$.

Performing the above construction for all sign or  Boolean gates $v$ in ${\mathcal M}_{\ell}$, with $d(v)=\ell +1$,
we obtain a graph ${\mathcal M}_{\ell +1}$, which is equivalent to ${\mathcal M}_{\ell}$, hence by the inductive hypothesis
the condition (3) is satisfied.
Observe that $d({\mathcal M}_{\ell +1}) \le d({\mathcal M}_\ell)$ for each $\ell$.

This completes the inductive construction.
Let $\mathcal M := {\mathcal M}_{d(\mathcal N)}$.
Observe that $\mathcal M$ may not be an arithmetic network graph because it may contain hanging vertices.
A hanging vertex may be one of two types: the ones labelled by $\neg$ from the original network $\mathcal N$,
and new gates created by the construction which remaind unused.

Let $w$ be the {\em last} Boolean or sign gate of the $\mathcal M$, i.e., the (unique) gate
whose unique child is the output gate.
Such gate exists since the last Boolean or sign gate $v$ exists in $\mathcal N$, and, according to the construction,
it either remains the last in $\mathcal M$, or $v$ is a Boolean gate labelled by $\neg$, and its (unique)
outgoing edge was re-attached to another gate, $a'$.
Remove from $\mathcal M$ all gates that are not ancestors of $w$, and denote the result by ${\mathcal N}'$.
This removes, in particular, all hanging vertices from $\mathcal M$, hence ${\mathcal N}'$ is an arithmetic network.
Let $F$ (respectively, $F'$) be the Boolean formula associated with $v$ in $\mathcal N$ (respectively,
$w$ in ${\mathcal N}'$).
By the construction, $F$ and $F'$ are equivalent Boolean formulae, thus they define the same semialgebraic set.
Therefore, $\mathcal N$ and ${\mathcal N}'$ are equivalent.

By the construction, $d({\mathcal M}_{\ell +1}) \le d({\mathcal M}_\ell)$ for each $\ell$, hence
$d({\mathcal M}) \le d({\mathcal N})$.
It follows that $d({\mathcal N}') \le d({\mathcal N})$, since obviously $d({\mathcal N}') \le d({\mathcal M})$.
\end{proof}

\begin{example}\label{ex:neg}
Consider an application of Lemma~\ref{le:negation} to a concrete network shown on Fig.~1.
The directed graph $\mathcal M$ is drawn on Fig.~2.
Note that $\mathcal M$ is {\em not} a network (it has hanging Boolean gates).
Fig.~3 shows the resulting network ${\mathcal N}'$.
\end{example}

\begin{figure}[hbt]
       \centerline{
          \scalebox{0.3}{
             \includegraphics{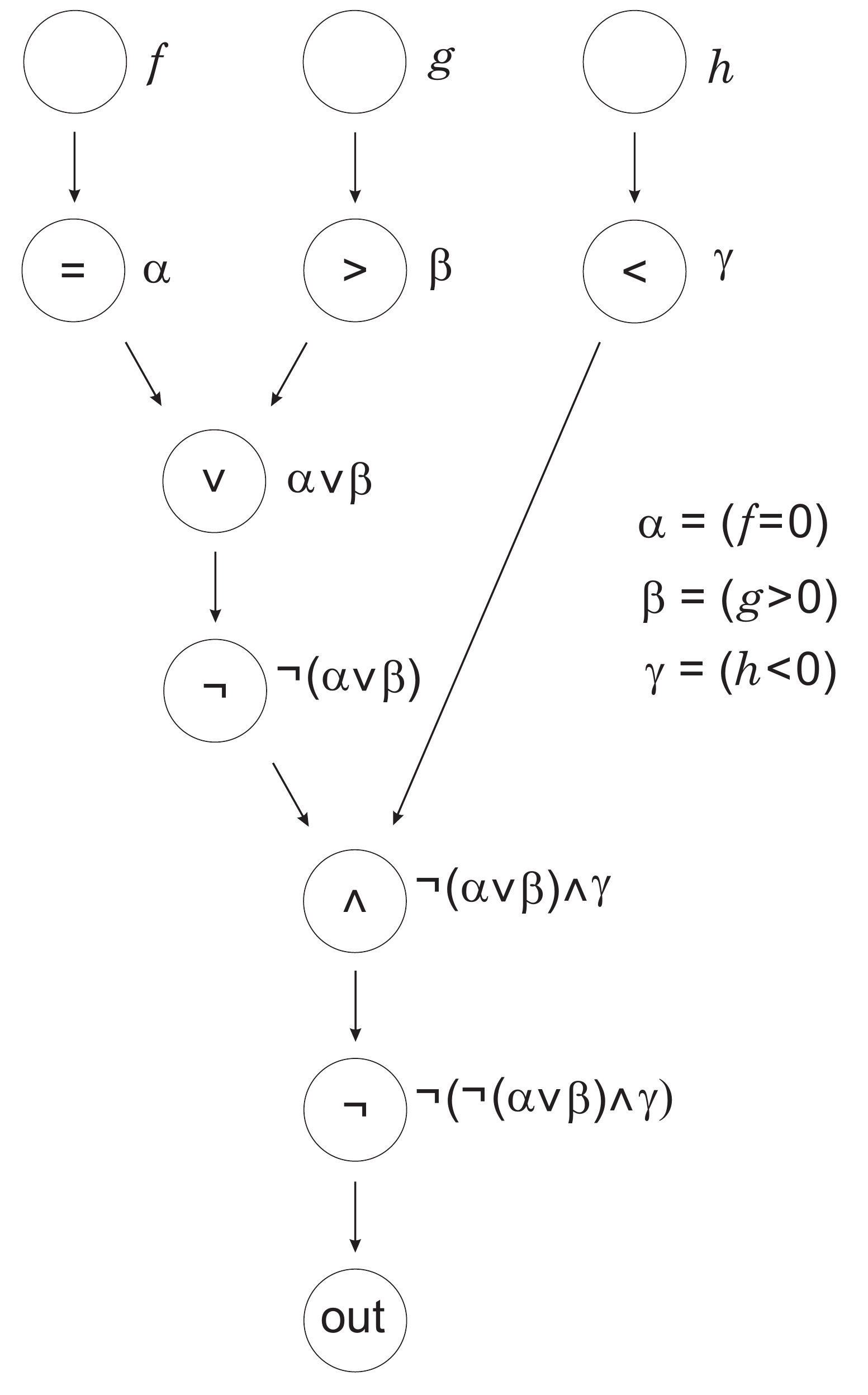}
             }
           }
\caption{ }
\label{fig:network1}
\end{figure}

\begin{figure}[hbt]
       \centerline{
          \scalebox{0.4}{
             \includegraphics{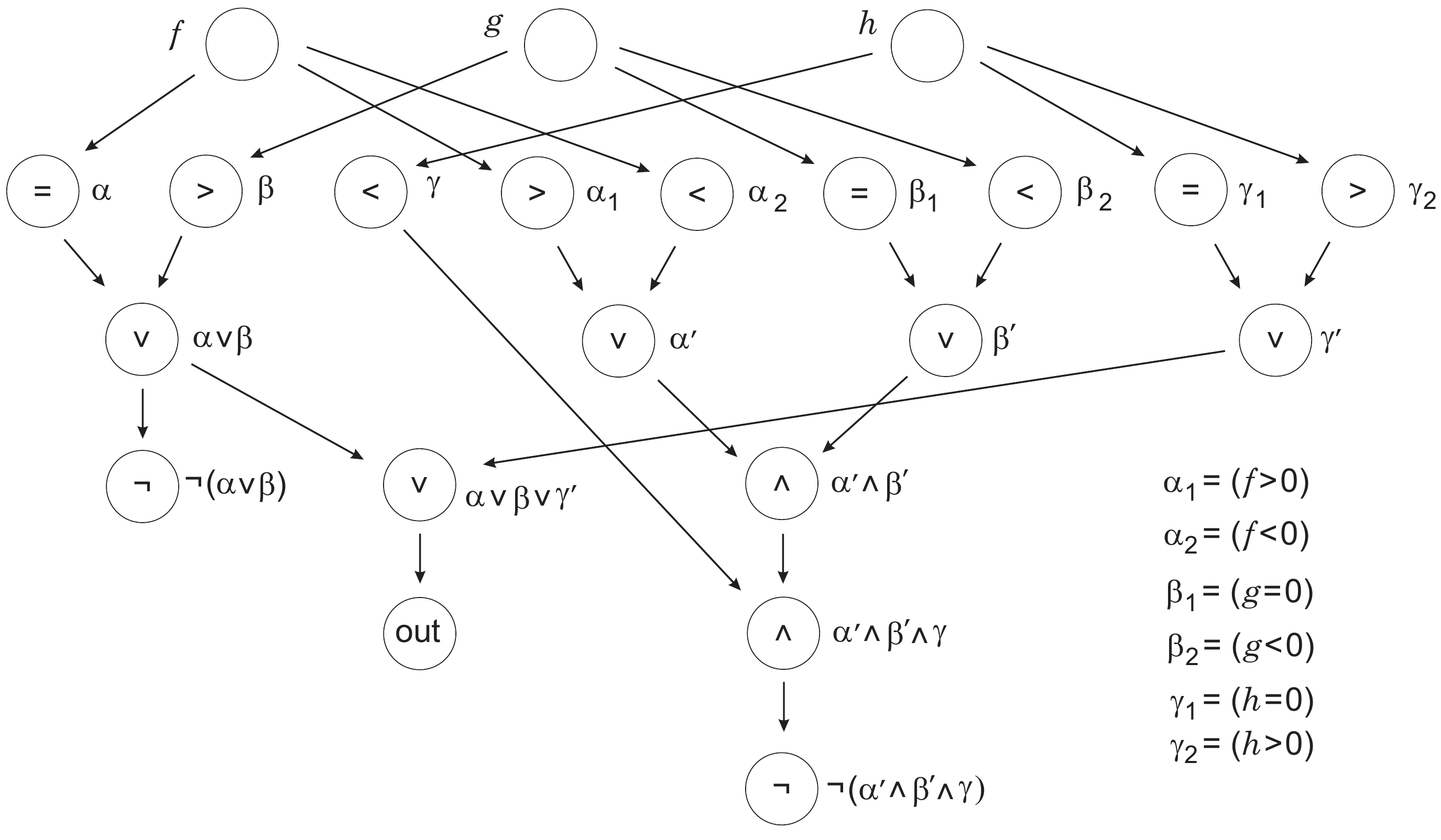}
             }
           }
\caption{ }
\label{fig:network1}
\end{figure}

\begin{figure}[hbt]
       \centerline{
          \scalebox{0.3}{
             \includegraphics{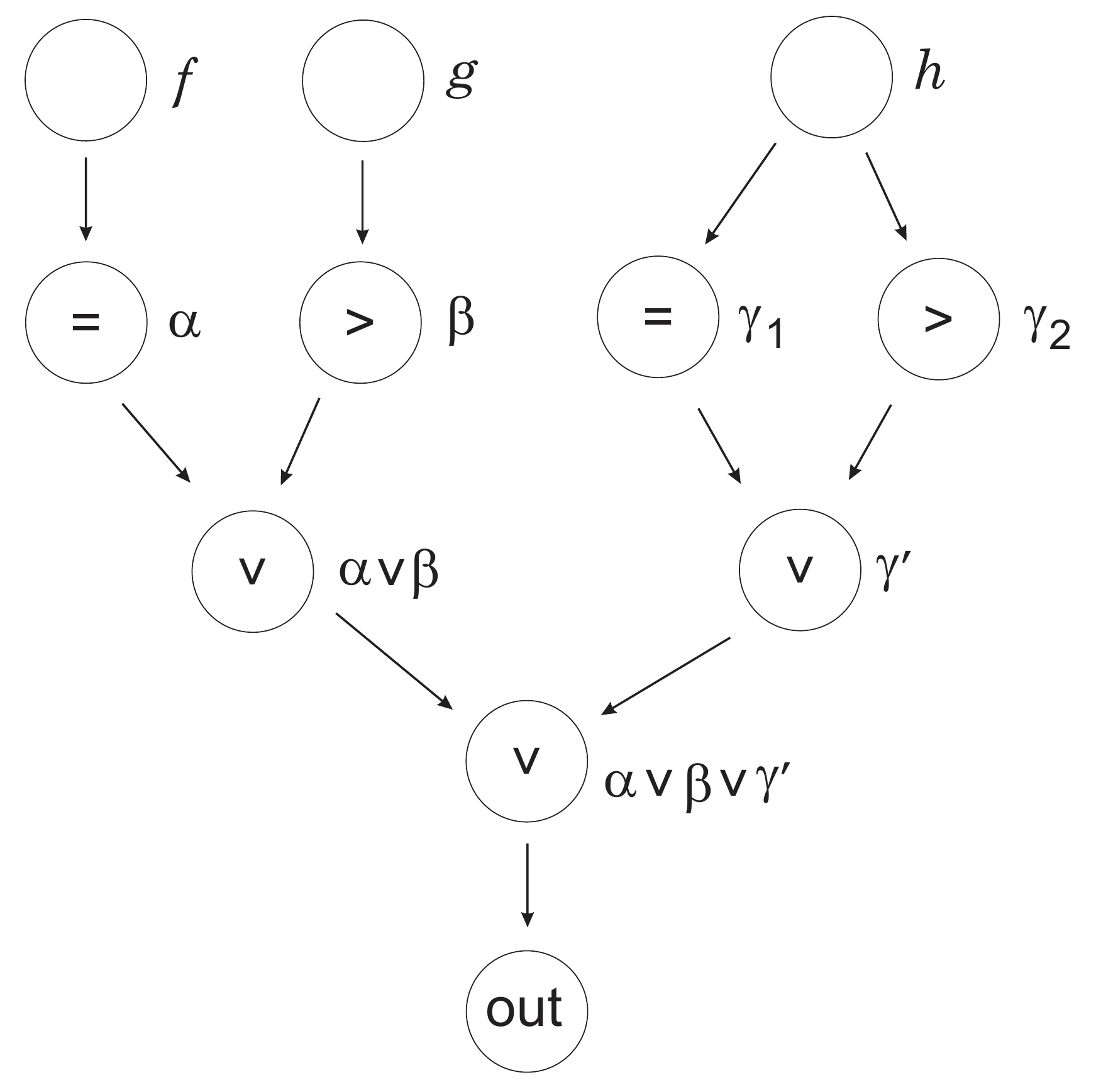}
             }
           }
\caption{ }
\label{fig:network1}
\end{figure}

\subsection{Modification of selection gates}\label{sec:mod}

Let $\mathcal N$ be an arithmetic network testing membership to a semi-algebraic set $\Sigma$.
By Lemma~\ref{le:negation}, we can assume that it has no gates labelled by $\neg$.
In Section~\ref{sec:lower} we will need to modify $\mathcal N$ so that the modified network tests membership in
a compact semi-algebraic set which is homotopy equivalent to $\Sigma$.
In this process, the modification of a Boolean formula, associated with a gate, and the modification
of its negation will not become negations of one another.

For instance (see details in Example~\ref{ex:boolean}), $f^2=0$ is equivalent to $\neg (f^2>0)$.
After compactification, $f^2=0$ will turn into $f^2 \le \eps$ while $f^2>0$ will turn into $f^2 \ge \delta$,
where $\eps$ and $\delta$ are some small positive real numbers with $\eps < \delta$.
At the same time, a selection gate having $f^2 \le \eps$ as its Boolean input, will automatically produce the
implicit complement condition $f^2 > \eps$, which is different from the required $f^2 \ge \delta$.

Thus we need to separate, for every selection gate, the Boolean formula, associated with its Boolean parent, from
the implicit negation of this formula.
To prepare this operation we now describe a further modification of $\mathcal N$ which results in another equivalent network,
${\mathcal N}''$, in which each selection gate is coupled with another selection gate, having the contrary Boolean parent.

As in the proof of Lemma~\ref{le:negation}, we can modify $\mathcal N$ so that for each selection gate with
a sign or a Boolean parent which has an associated Boolean formula $A$, the resulting directed graph simultaneously
has a sign or a Boolean gate with an associated Boolean formula $A'$ which is equivalent to $\neg A$.

More precisely, let $\mathcal M$ be the directed graph defined in the proof of Lemma~\ref{le:negation}.
For every selection gate in $\mathcal M$ do the following.
Replace in $\mathcal M$ the subgraph of the type shown on Fig.~4 by another subgraph, shown on Fig.~5.
(Note that for this we may need to introduce an additional {\em constant gate} labelled by $0$.)
It is easy to observe that the arithmetic $(+)$-gate in Fig.~5 outputs the same numerical value as the selection
gate of $\mathcal N$ in Fig.~4.

\begin{figure}[hbt]
       \centerline{
          \scalebox{0.3}{
             \includegraphics{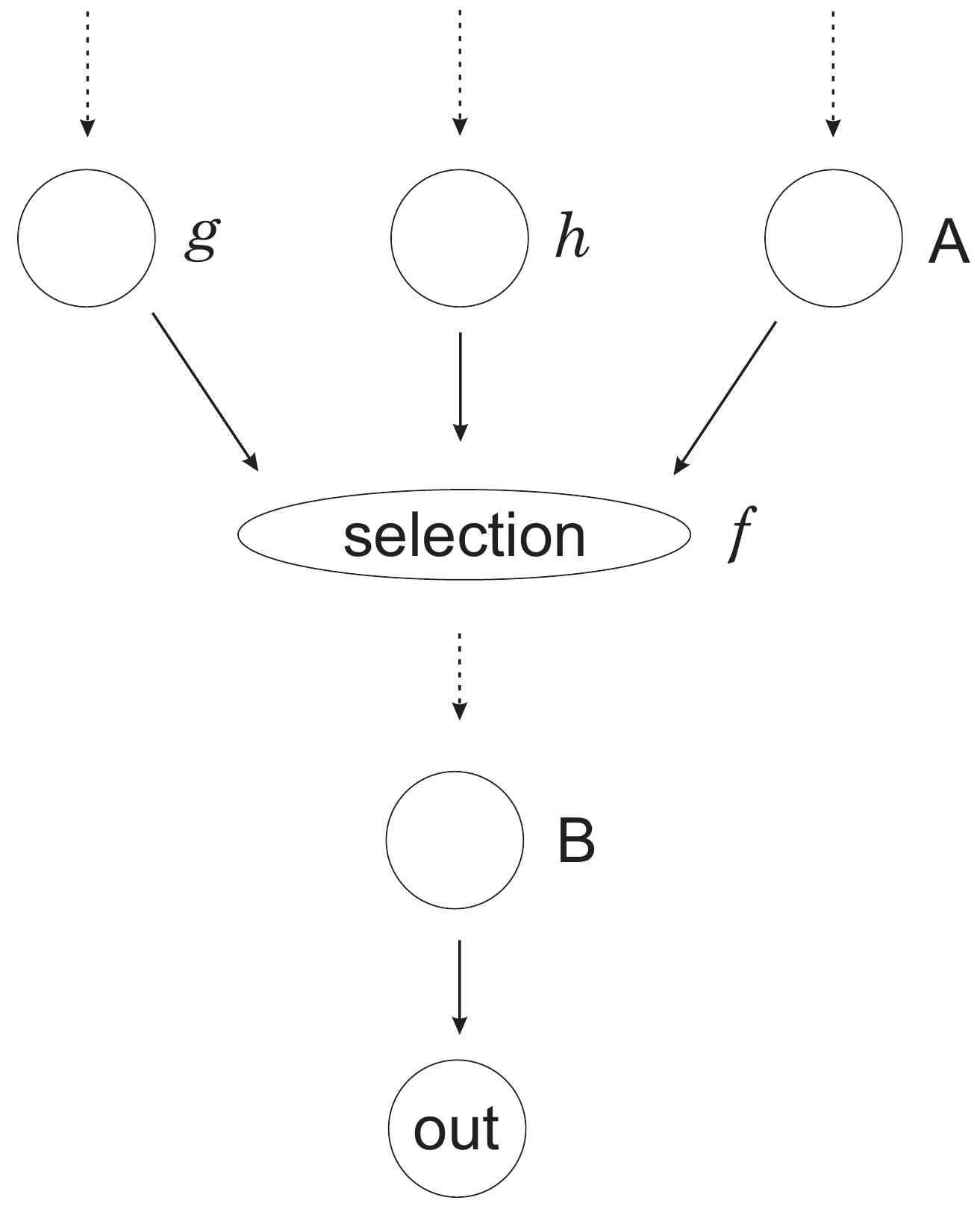}
             }
           }
\caption{ }
\label{fig:network1}
\end{figure}

The output gate in Fig.~4 has associated Boolean formula $B$.
Since $A'$ is equivalent to $\neg A$, the Boolean formula $A \lor A'$ is identically true,
hence the truth value of the Boolean formula $B \land (A \lor A')$, associated with the output gate in
Fig.~5, is always the same as the truth value of $B$.

\begin{figure}[hbt]
       \centerline{
          \scalebox{0.3}{
             \includegraphics{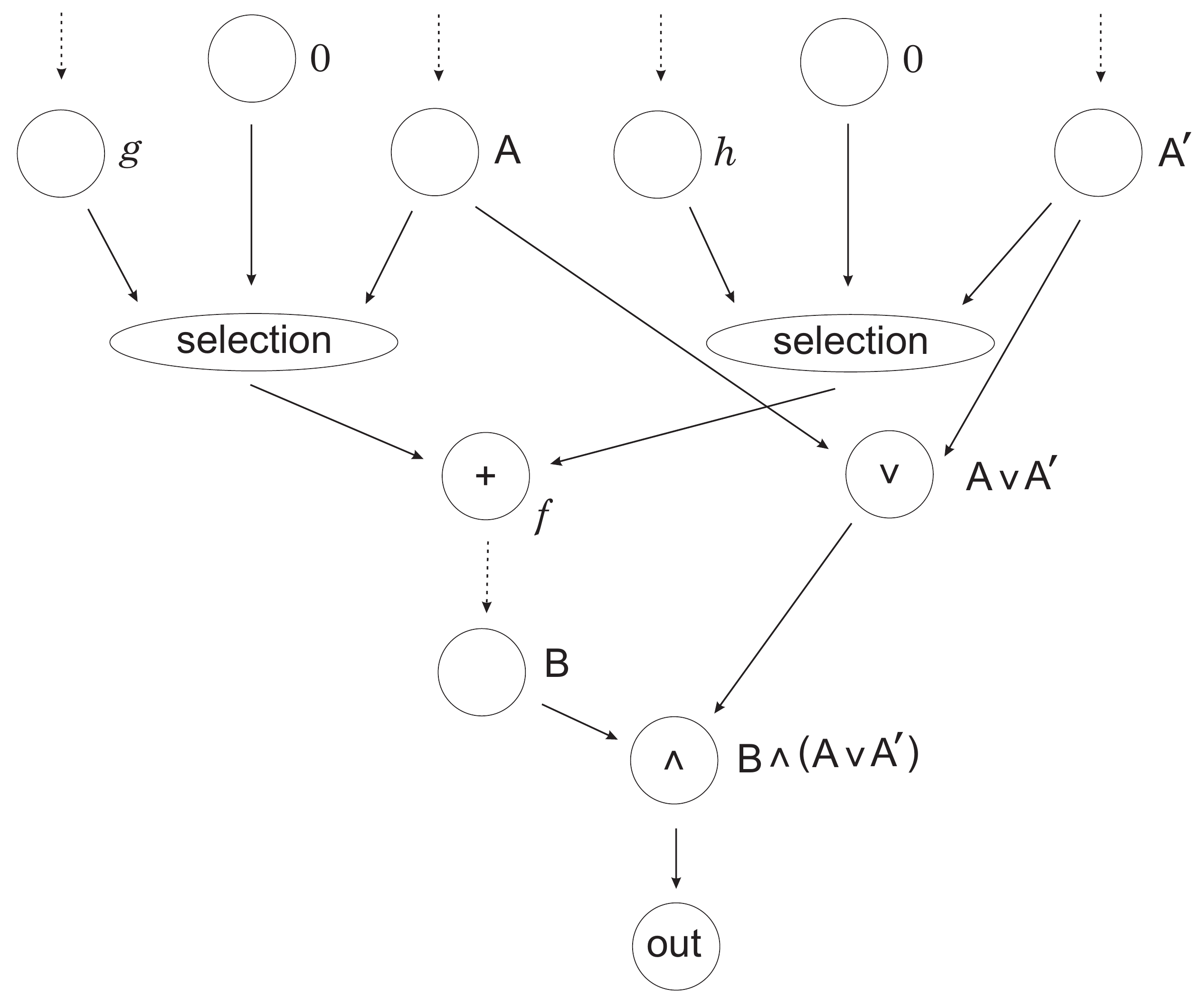}
             }
           }
\caption{ }
\label{fig:network1}
\end{figure}

Once the replacement is done for every selection gate in $\mathcal M$, we take the conjunction
of Boolean formulae associated with of all output gates in Fig.~5, using a dichotomy (binary tree).
Remove from the resulting graph all gates that are not ancestors of the last gate.
Denote the obtained network by ${\mathcal N}''$.

We've just proved the first half of the following lemma.

\begin{lemma}\label{le:selproc}
Networks ${\mathcal N}''$ and $\mathcal N$ are equivalent, and $d({\mathcal N}'')=O(d(\mathcal N))$.
\end{lemma}

\begin{proof}
The equivalence of ${\mathcal N}''$ and $\mathcal N$ has been proved along with the construction of ${\mathcal N}''$.

The only reason why the depth of ${\mathcal N}''$ may increase from $d(\mathcal N)$, is the necessity to take a conjunction
of Boolean formulae associated with all output gates in subgraphs of the type shown on Fig.~5 by a dichotomy,
using the additional depth not greater than the logarithm of the number of selection gates in $\mathcal N$.
Since the number of selection gates does not exceed $3^{d({\mathcal N})+1}$ (the total number of gates),
the depth will increase by $O(d(\mathcal N))$, thus the total depth of ${\mathcal N}''$ becomes $O(d(\mathcal N))$.
\end{proof}

\section{Topological tools}

This section coincides (up to minor details) with Section~2 in \cite{GV14}, and is reproduced for reader's
convenience.
Here we formulate some results from \cite{GV05, GV09, GVZ} which are used further in this paper.

In what follows, for a topological space $X$, let ${\rm b}_m(X):={\rm rank}\ H_m(X)$ be its $m$-th Betti number
with respect to the singular homology group $H_m(X)$ with coefficients in some fixed Abelian group.
By ${\rm b}(X)$ we denote the {\em total} Betti number of $X$, i.e., the sum
$\sum_{i \ge 0} {\rm b}_i (X)$.

\subsection{Approximation by monotone families}\label{sub:approx}
\begin{definition}\label{def:S_delta}
Let $G$ be a compact semialgebraic set.
Consider a semialgebraic family $\{ S_\delta \}_{\delta >0}$ of
compact subsets of $G$, such that for all $\delta', \delta \in (0,1)$,
if $\delta' > \delta$, then $S_{\delta'} \subset S_{\delta}$.
Denote $S := \bigcup_{\delta >0} S_{\delta}$.

For each $\delta >0$, let $\{ S_{\delta, \eps} \}$ be a semialgebraic family
of compact subsets of $G$ such that:
\begin{itemize}
\item[(i)]
for all $\eps, \eps' \in (0,1)$, if $\eps' > \eps$, then
$S_{ \delta, \eps} \subset S_{ \delta, \eps'}$;
\item[(ii)]
$S_{\delta}= \bigcap_{\eps >0} S_{\delta, \eps}$;
\item[(iii)]
for all $\delta' >0$ sufficiently smaller than
$\delta$, and for all $\eps' >0$, there exists an open in $G$ set $U \subset G$
such that $S_{\delta} \subset U \subset S_{\delta' , \eps'}$.
\end{itemize}
We say that $S$ is {\em represented} by the families $\{ S_\delta \}$ and
$\{ S_{\delta, \eps} \}$ in $G$.
\end{definition}

Consider the following two particular cases.
\medskip

\noindent {\bf Case 1.}\quad Let a semialgebraic set $S$ be given
by a Boolean formula {\em with no negations}, and with atomic subformulae of the kind $f=0$ or $f>0$.
Let $\delta$ and $\eps$ be some positive constants.

Suppose first that $S$ is bounded in $\Real^n$, take as $G$ a closed ball of a sufficiently large
radius centered at 0.
The set $S_{\delta}$ is the result of the replacement of all inequalities $h>0$ and $h<0$ by $h \ge \delta$
and $h \le -\delta$ respectively.
The set $S_{\delta, \eps}$ is obtained by replacing
all expressions $h>0$, $h<0$ and $h=0$ by $h \ge \delta$, $h \le -\delta$ and
$h^2 - \eps \le 0$, respectively.
By Lemma~1.2 in \cite{GV09}, the set $S$, is {\em represented} by families $\{ S_{\delta} \}$ and
$\{ S_{\delta, \eps} \}$ in $G$.

Now suppose that $S$ is not necessarily bounded.
In this case as $G$ take the semialgebraic one-point (Alexandrov) compactification of $\Real^n$.
Define sets $S_\delta$ and $S_{\delta, \eps}$ as in the bounded case, replacing equations and
inequalities, and then taking the conjunction of the resulting formula with $|\x|^2 \le 1/\delta$.
Again, $S$ is represented by $\{ S_{\delta} \}$ and $\{ S_{\delta, \eps} \}$ in $G$.
\medskip

\noindent {\bf Case 2.}\quad Let $\rho:\> \Real^{n+r} \to \Real^n$ be the projection map, and
$S \subset \Real^{n+r}$ be a semialgebraic set, given as a disjoint union of {\em basic} semialgebraic sets.
The set $S$ is represented by families $\{ S_\delta \}$,
$\{ S_{\delta , \eps} \}$ in the compactification of $\Real^{n+r}$ as
described in {\bf Case~1}.
One can verify \cite{GV09}, that the projection $\rho (S)$ is represented by families
$\{ \rho(S_\delta) \}$, $\{ \rho(S_{\delta , \eps}) \}$ in the Alexandrov compactification of $\Real^n$.
\medskip

Returning to the general case, suppose that a semialgebraic set $S$ is {\em represented} by families
$\{ S_{\delta} \}$ and $\{ S_{\delta, \eps} \}$ in $G$.

For a sequence $\eps_0 , \delta_0 ,\eps_1 , \delta_1 , \ldots ,\eps_m , \delta_m$,
where $m \ge 0$, introduce the compact set
$$T_m(S):=S_{\delta_0,\eps_0}\cup S_{\delta_1,\eps_1} \cup \cdots \cup S_{\delta_m,\eps_m}.$$

Observe that in {\bf Case~2}, we have the equality
\begin{equation}\label{eq:rho}
T_m(\rho(S))=\rho(T_m(S)).
\end{equation}

In what follows, for two real numbers $a,\ b$ we write $a \ll b$ to
``mean $a$ is sufficiently smaller than $b$'' (see formal Definition~1.7 in \cite{GV09}).

\begin{proposition}[\cite{GV09}, Theorem~1.5]\label{pr:main}
For any $m \ge 0$, and
$$0<\eps_0 \ll\delta_0\ll\eps_1 \ll \delta_1 \ll \cdots \ll\eps_m \ll \delta_m \ll 1$$
we have
\begin{itemize}
\item[(i)]
for every $1 \le k \le m$, there is an epimorphism $\varphi_k:\> H_k(T_m(S)) \to H_k(S)$,
in particular, ${\rm b}_k(S) \le {\rm b}_k(T_m(S))$;
\item[(ii)]
in {\bf Case~1}, for every $1 \le k \le m-1$, the epimorphism $\varphi_k$ is an
isomorphism, in particular, ${\rm b}_k(S) = {\rm b}_k(T_m(S))$.
Moreover, if $m \ge \dim (S)$, then $T_m(S)$ is homotopy equivalent to $S$.
\end{itemize}
\end{proposition}

\subsection{Betti numbers of projections}

\begin{definition}
For two maps $f_1:\> X_1 \to Y$ and $f_2:\> X_2 \to Y$ , the {\em fibered product} of $X_1$ and $X_2$
is defined as
$$X_1 \times_Y X_2:=\{ (\x_1,\x_2)\in X_1 \times X_2|\> f_1(\x_1)=f_2(\x_2)\}.$$
\end{definition}

\begin{proposition}[\cite{GVZ}, Theorem~1]\label{pr:map}
Let $f:\> X \to Y$ be a closed surjective semialgebraic map (in particular, $f$ can be the projection
map to a subspace, with a compact $X$).
Then
$${\rm b}_m(Y) \le \sum_{p+q=m} {\rm b}_q(W_p),$$
where
$$W_p:= \underbrace{X \times_Y \cdots \times_Y X}_\text{(p+1) {\rm times}}.$$
\end{proposition}

\section{Lower bounds}\label{sec:lower}

\begin{theorem}\label{th:general}
Let $\mathcal N$ be an arithmetic network testing membership in a semi-algebraic set $\Sigma \subset \Real^n$.
Then
$$d({\mathcal N}) \ge c_1 \sqrt{ \frac{\log ({\rm b}(\Sigma))}{n}} -c_2 \log n,$$
where $c_1,\ c_2$ are some positive constants.
\end{theorem}

We first prove two auxiliary lemmas.

Let $\mathcal N$ be an arithmetic network, testing membership in a semi-algebraic set $\Sigma \subset \Real^n$.
Consider networks ${\mathcal N}'$ and ${\mathcal N}''$, defined in Section~\ref{sec:elim} and \ref{sec:mod}, which are
equivalent to $\mathcal N$.
By Lemma~\ref{le:selproc}, $d({\mathcal N}'')=O(d(\mathcal N))$, hence it is sufficient to prove Theorem~\ref{th:general}
for the network ${\mathcal N}''$.

Note that ${\mathcal N}''$ has no Boolean gates labelled by $\neg$.
According to Section~\ref{sec:denotational}, networks ${\mathcal N}'$ and ${\mathcal N}''$ have associated Boolean formulae
$B({\mathcal N}')$ and $B({\mathcal N}'')$ respectively, such that $\Sigma$ coincides with the set of all points in $\Real^n$
satisfying each of $B({\mathcal N}')$ and $B({\mathcal N}'')$.
In addition, these formulae have no $\neg$-symbols.
Observe that $B({\mathcal N}'')$ is of the kind $B({\mathcal N}') \land U$, where $U$ is the conjunction of an identically true
Boolean formulae of the type $A \lor A'$ from the construction of ${\mathcal N}''$ (see Fig.~3.1 (b)),
for all selection gates in ${\mathcal N}$.

Denote by $A_{\delta, \eps}$ (respectively, by $U_{\delta, \eps}$) the Boolean formula obtained from $B({\mathcal N}'')$
(respectively, from $U$) by the rules in Section~\ref{sub:approx}, Case~1, the subcase of a bounded set.
Let $\Sigma_{\delta, \eps}$ be the set of all points in $\Real^n$ satisfying $A_{\delta, \eps}$.

\begin{lemma}\label{le:deltaeps}
There exists an arithmetic network ${\mathcal N}_{\delta, \eps}$, testing membership in
$\Sigma_{\delta, \eps}$, such that $d({\mathcal N}_{\delta, \eps}) \le c_1 d({\mathcal N}'')+ c_2 \log n$ for
some positive constants $c_1,\ c_2$.
\end{lemma}

\begin{proof}
For all sign gates $v$ in ${\mathcal N}''$, taken in arbitrary order, do the following.
Let $v$ be labelled by a sign $\sigma \in \{ <,=,> \}$.
Its parent $w$ is either an input gate, or a constant gate, or an arithmetic gate, or a selection gate,
and has an associated piecewise polynomial $f$ defined on a partition.
Replace the directed edge $(w,v)$ by a directed graph, depending on
the sign $\sigma$, as shown on Fig.~6, 7 and 8.

Note that for this we may need to introduce up to two additional {\em constant gates} labelled by $\eps$ or
$\delta$, depending on $\sigma$.
It is easy to see that the result of all the replacements is an arithmetic network, denote it by
${\mathcal M}_{\delta, \eps}$, herewith $d({\mathcal M}_{\delta, \eps}) =O(d({\mathcal N}'')$.

There is the following injective map $L$ from the set of all gates of ${\mathcal N}''$ to the set of all gates
of ${\mathcal M}_{\delta, \eps}$.
All vertices in the graph of ${\mathcal N}''$ that are not replaced in the construction of
${\mathcal M}_{\delta, \eps}$ (i.e., all gates except sign gates) are mapped identically.
Sign gates of ${\mathcal N}''$ with associated Boolean formulae $f=0,\ f>0$ and $f<0$ are mapped to Boolean
$\lor$-gates of ${\mathcal M}_{\delta, \eps}$ with associated Boolean formulae $(f^2 -\eps=0) \lor (f^2- \eps=0)$,
$(f - \delta>0) \lor (f- \delta=0)$ and $(-f- \delta>0) \lor (-f- \delta=0)$ respectively.

\begin{figure}[hbt]
       \centerline{
          \scalebox{0.3}{
             \includegraphics{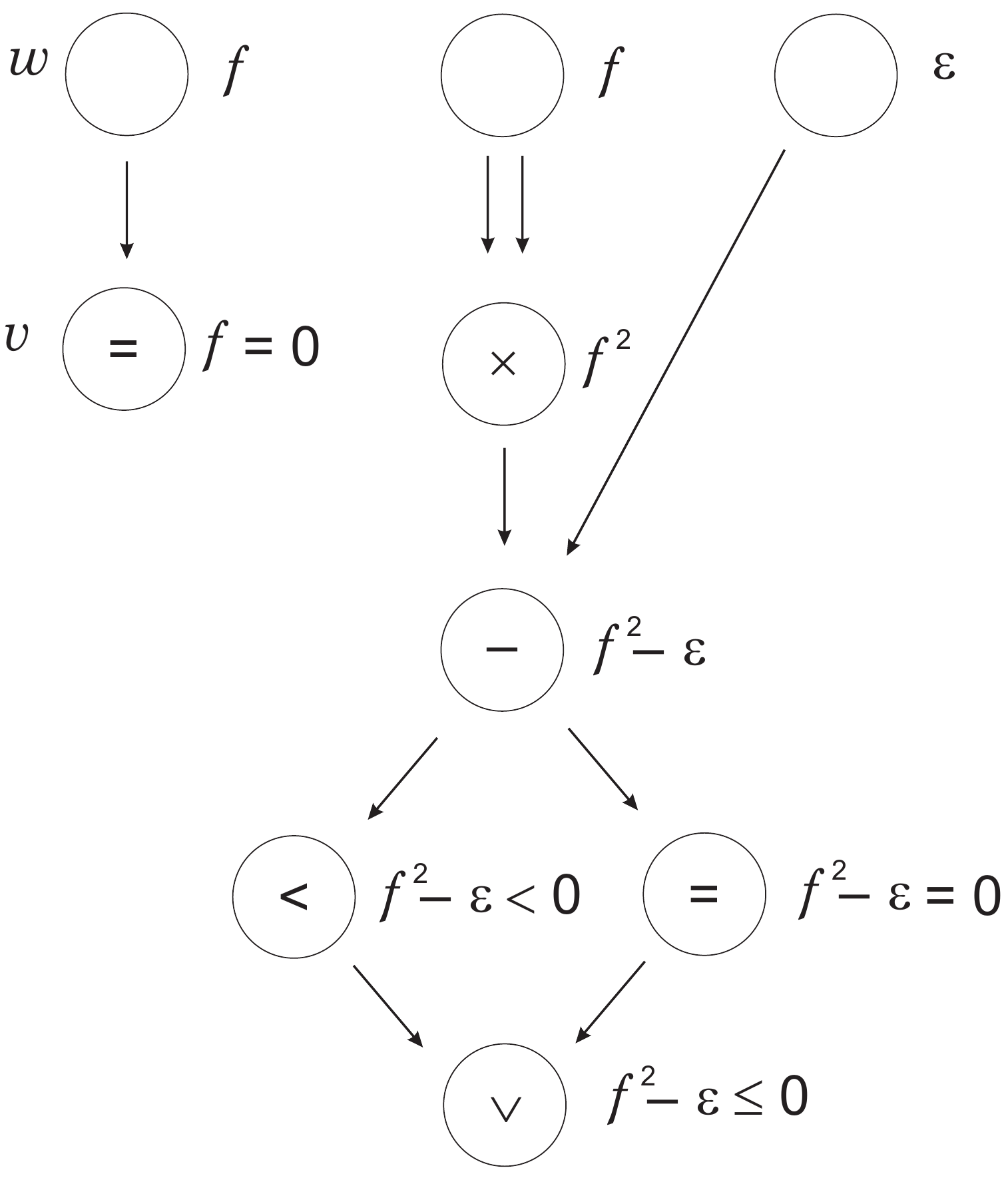}
             }
           }
\caption{ }
\label{fig:network1}
\end{figure}

\begin{figure}[hbt]
       \centerline{
          \scalebox{0.3}{
             \includegraphics{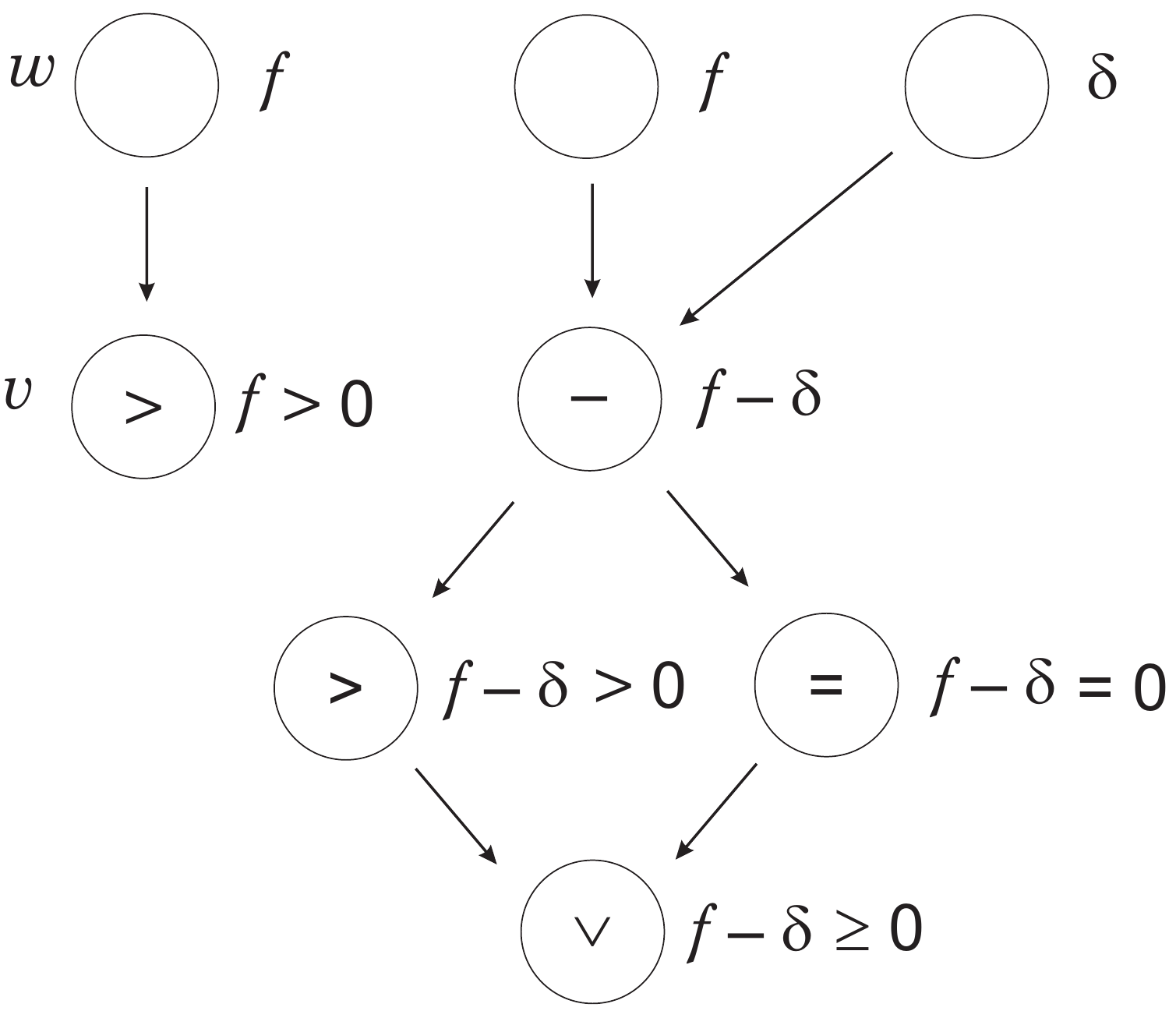}
             }
           }
\caption{ }
\label{fig:network1}
\end{figure}

\begin{figure}[hbt]
       \centerline{
          \scalebox{0.3}{
             \includegraphics{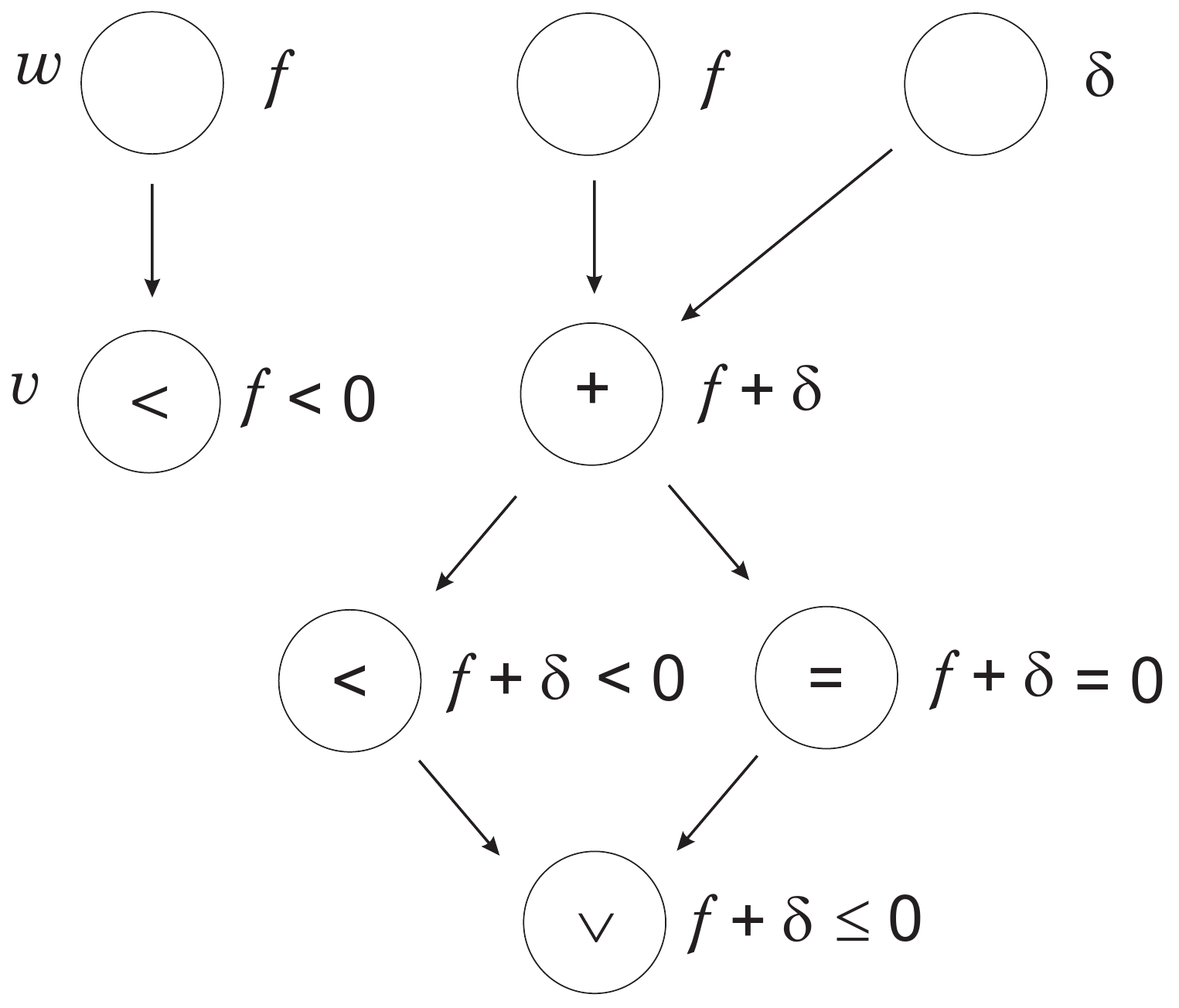}
             }
           }
\caption{ }
\label{fig:network1}
\end{figure}

Let $B({\mathcal M}_{\delta, \eps})$ be the Boolean formula associated with ${\mathcal M}_{\delta, \eps}$.
We prove that $A_{\delta, \eps}$ is equivalent to $B({\mathcal M}_{\delta, \eps})$:
$$
\begin{CD}
{\mathcal N}'' @>\text{Fig.~6, 7, 8}>> {\mathcal M}_{\delta, \eps}\\
@V\text{Section~\ref{sec:denotational}}VV  @.\\
B({\mathcal N}'') @. @V\text{Section~\ref{sec:denotational}}VV\\
@V\text{Section~\ref{sub:approx}}VV @.\\
A_{\delta, \eps} @>\text{equivalent}>> B({\mathcal M}_{\delta, \eps})
\end{CD}
$$

Note that $A_{\delta, \eps}$ and $B({\mathcal M}_{\delta, \eps})$ are not necessarily identical
(see Example~\ref{ex:boolean} below).

Consider a gate $v$ in ${\mathcal N}''$ having the depth $\ell$.
Depending on the type of $v$, it has either an associated Boolean formula $B$ or an associated piecewise polynomial $f$,
represented by a list $F:=f_1, \ldots ,f_k; B_1, \ldots ,B_k$ of polynomials and Boolean formulae for
elements of the partition.
In the case of a Boolean formula, let $B_{\delta, \eps}$ be the Boolean formula obtained from $B$
by the rules in Section~\ref{sub:approx}, Case~1.
In the case of a piecewise polynomial, let $F_{\delta, \eps}$ be the list
$f_1, \ldots ,f_k; B_{1,\delta, \eps}, \ldots ,B_{k,\delta, \eps}$, where for each $i$,
$B_{i, \delta, \eps}$ is the Boolean formula obtained from $B_i$
by the rules in Section~\ref{sub:approx}, Case~1.
Notice that sets $B_{1,\delta, \eps}, \ldots ,B_{k,\delta, \eps}$ are pairwise disjoint, hence
$F_{\delta, \eps}$ represents a {\em partial} piecewise polynomial $\overline f$ defined in their union.
We prove by induction on $\ell$ that, depending on the type of $v$, the gate $L(v)$ in
${\mathcal M}_{\delta, \eps}$ either has the associated Boolean formula $\widehat B_{\delta, \eps}$ such that
$\widehat B_{\delta, \eps} \land U_{\delta, \eps}$ is equivalent to
$B_{\delta, \eps} \land U_{\delta, \eps}$, or has the associated piecewise polynomial $\widehat f$, defined by a list
$\widehat F_{\delta, \eps}$, such that its restriction $\widehat f|_{U_{\delta, \eps}}$ coincides with
$\overline f|_{U_{\delta, \eps}}$.

The base of induction, for $\ell =0$, is trivial.

On the inductive step, let $v$ be a sign gate with an associated Boolean formula $f\ \sigma\ 0$, where $f$ is
a piecewise polynomial.
Let, for definiteness, $\sigma$ be $=$.
Boolean formula $B_{\delta, \eps}$ is then of the form $(\overline f)^2 - \eps \le 0$,
where $\overline f$ is a partial piecewise polynomial.
Observe that $L(v)$ is a Boolean gate in ${\mathcal M}_{\delta, \eps}$ with the associated formula $\widehat B_{\delta, \eps}$
of the form $((\widehat f)^2- \eps<0) \lor ((\widehat f)^2- \eps=0)$ where $\widehat f$ is a piecewise polynomial.
By the inductive hypothesis, $\widehat f|_{U_{\delta, \eps}}$ coincides with $\overline f|_{U_{\delta, \eps}}$.
Hence $\widehat B_{\delta, \eps} \land U_{\delta, \eps}$ is equivalent to $B_{\delta, \eps} \land U_{\delta, \eps}$.
This completes the inductive step for a sign gate $v$.

Cases when $v$ is an arithmetic or Boolean gate are analogous.

Let $v$ and $v'$ be some paired selection gates of ${\mathcal N}''$, and $w,\ w'$ be their respective Boolean parents.
Let the arithmetic input of $v$ (respectively, of $v'$) be $(g,0)$ (respectively, $(h,0)$).
Here $g$ and $h$ are piecewise polynomials defined by some lists
$g_1, \ldots ,g_\alpha; G_1, \ldots ,G_\alpha$ and $h_1, \ldots ,h_\beta; H_1, \ldots ,H_\beta$ respectively.
Let $C$ (respectively, $C'$) be the Boolean formula associated with $w$ (respectively, with $w'$).
Both $C$ and $C'$ don't have the $\neg$-symbol, and one is equivalent to the negation of another.
Let for each $i$ Boolean formulae $G_{i,\delta, \eps}$ and $H_{i, \delta, \eps}$ be obtained from $G_i$ and $H_i$
respectively by the rules in Section~\ref{sub:approx}.
Let $C_{\delta, \eps}$ and $C'_{\delta, \eps}$ be Boolean formulae obtained from $C$ and $C'$ respectively
by the rules in Section~\ref{sub:approx}.

The arithmetic output of the pair $v,\ v'$ is the piecewise polynomial $f$ represented by the list
$$F:=g_1, \ldots ,g_\alpha, h_1, \ldots ,h_\beta; G_1 \land C, \ldots ,G_\alpha \land C,
H_1 \land C', \ldots ,H_\beta \land C'.$$
Hence the corresponding list $F_{\delta, \eps}$ for $\overline f$ is
$$g_1, \ldots ,g_\alpha, h_1, \ldots ,h_\beta; G_{1, \delta, \eps} \land C_{\delta, \eps}, \ldots ,
G_{\alpha, \delta, \eps} \land C_{\delta, \eps}, H_{1, \delta, \eps} \land C'_{\delta, \eps}, \ldots ,
H_{\beta, \delta, \eps} \land C'_{\delta, \eps}.$$
Note that the disjunction $(C_{\delta, \eps} \lor C'_{\delta, \eps})$ will appear as a conjunction member in
the Boolean formula $A_{\delta, \eps}$.

On the other hand, for the selection gates $L(v)=v$, $L(v')=v'$ in ${\mathcal M}_{\delta, \eps}$
the arithmetic output of the pair $v,\ v'$ is the piecewise polynomial $\widehat f$ represented by the list
$$\widehat F_{\delta, \eps}:= g_1, \ldots ,g_\alpha, h_1, \ldots ,h_\beta, 0;$$
$$\widehat G_{1, \delta, \eps} \land \widehat C_{\delta, \eps}, \ldots ,
\widehat G_{\alpha, \delta, \eps} \land \widehat C_{\delta, \eps}, \widehat H_{1, \delta, \eps} \land
\widehat C'_{\delta, \eps}, \ldots , \widehat H_{\beta, \delta, \eps} \land \widehat C'_{\delta, \eps},
\neg (\widehat C_{\delta, \eps} \lor \widehat C'_{\delta, \eps}).$$
By the inductive hypothesis, $\widehat G_{i, \delta, \eps} \land U_{\delta, \eps}$ is equivalent to
$G_{i, \delta, \eps} \land U_{\delta, \eps}$,
$\widehat H_{i, \delta, \eps} \land U_{\delta, \eps}$ is equivalent to $H_{i, \delta, \eps} \land U_{\delta, \eps}$
for every $i$, $\widehat C_{\delta, \eps} \land U_{\delta, \eps}$ is equivalent to
$C_{\delta, \eps} \land U_{\delta, \eps}$, and
$\widehat C'_{\delta, \eps} \land U_{\delta, \eps}$ is equivalent to $C'_{\delta, \eps} \land U_{\delta, \eps}$.

Observe that the Boolean formula $U_{\delta, \eps}$ is a conjunction, with one of its members being
$(C_{\delta, \eps} \lor C'_{\delta, \eps})$.
Taking conjunctions of Boolean formulae in both lists, $F_{\delta, \eps}$ and $\widehat F_{\delta, \eps}$,
with $U_{\delta, \eps}$, we get, in particular, an empty set defined by the formula
$$\neg (\widehat C_{\delta, \eps} \lor \widehat C'_{\delta, \eps}) \land
(C_{\delta, \eps} \lor C'_{\delta, \eps}).$$
Thus, the restrictions of $\overline f$ and $\widehat f$ to $U_{\delta, \eps}$ coincide.
This completes the induction step in the case of a selection gate $v$.

On the last induction step we get the equivalence of the sets $A_{\delta, \eps} \land U_{\delta, \eps}=A_{\delta, \eps}$
and $B({\mathcal M}_{\delta, \eps}) \land U_{\delta, \eps}=B({\mathcal M}_{\delta, \eps})$.

To complete the construction of ${\mathcal N}_{\eps, \delta}$, we need to add the condition
$$|(x_1, \ldots ,x_n)|^2 \le 1/\delta$$
for the input vector $(x_1, \ldots ,x_n)$.
Using the dichotomy, we compute the sum $x_1^2+ \cdots +x_n^2$ with the depth $O(\log n)$.
Squaring and comparing with $\/ \delta$ requires additional constant depth.
Attach the resulting graph to ${\mathcal M}_{\delta, \eps}$ in a straightforward way.
The result is the sought network ${\mathcal N}_{\eps, \delta}$.
\end{proof}

The following example illustrates Lemma~\ref{le:deltaeps} and its proof.

\begin{example}\label{ex:boolean}
Consider the network $\mathcal N$ on Fig~9.
Clearly,
$$
B({\mathcal N})=((g=0 \land f^2>0) \lor (h=0 \land f=0)).
$$
A network obtained from it by means of Lemma~\ref{le:selproc} is shown on Fig.~10, denote it by ${\mathcal N}''$.
In this case,
$$
B({\mathcal N}'')=((g=0 \land f^2>0) \lor (h=0 \land f=0)) \land (f^2>0  \lor f=0)
$$
and
$$
A_{\delta, \eps}=((g^2 \le \eps \land f^2 \ge \delta) \lor (h^2 \le \eps \land f^2 \le \eps))
\land (f^2 \ge \delta \lor f^2 \le \eps).
$$
Here the Boolean formula $U$ is $f^2>0  \lor f=0$, while $U_{\delta, \eps}$ is $f^2 \ge \delta \lor f^2 \le \eps$.

The network ${\mathcal M}_{\delta, \eps}$ is not shown on a picture, but one can write out
$$
B({\mathcal M}_{\delta, \eps})=((g^2 \le \eps \land f^2 \ge \delta) \lor (h^2 \le \eps \land f^2 \le \eps)
\lor (0 \le \eps \land f^2 < \delta \land f^2 > \eps)) \land
$$
$$
\land (f^2 \ge \delta \lor f^2 \le \eps).
$$
The last term, $f^2 \ge \delta \lor f^2 \le \eps$, of $B({\mathcal M}_{\delta, \eps})$ kills the
third term, $0 \le \eps \land f^2 < \delta \land f^2 > \eps$, in the first disjunction, thus
$B({\mathcal M}_{\delta, \eps})$ is equivalent to $A_{\delta, \eps}$.
\end{example}

\begin{figure}[hbt]
       \centerline{
          \scalebox{0.3}{
             \includegraphics{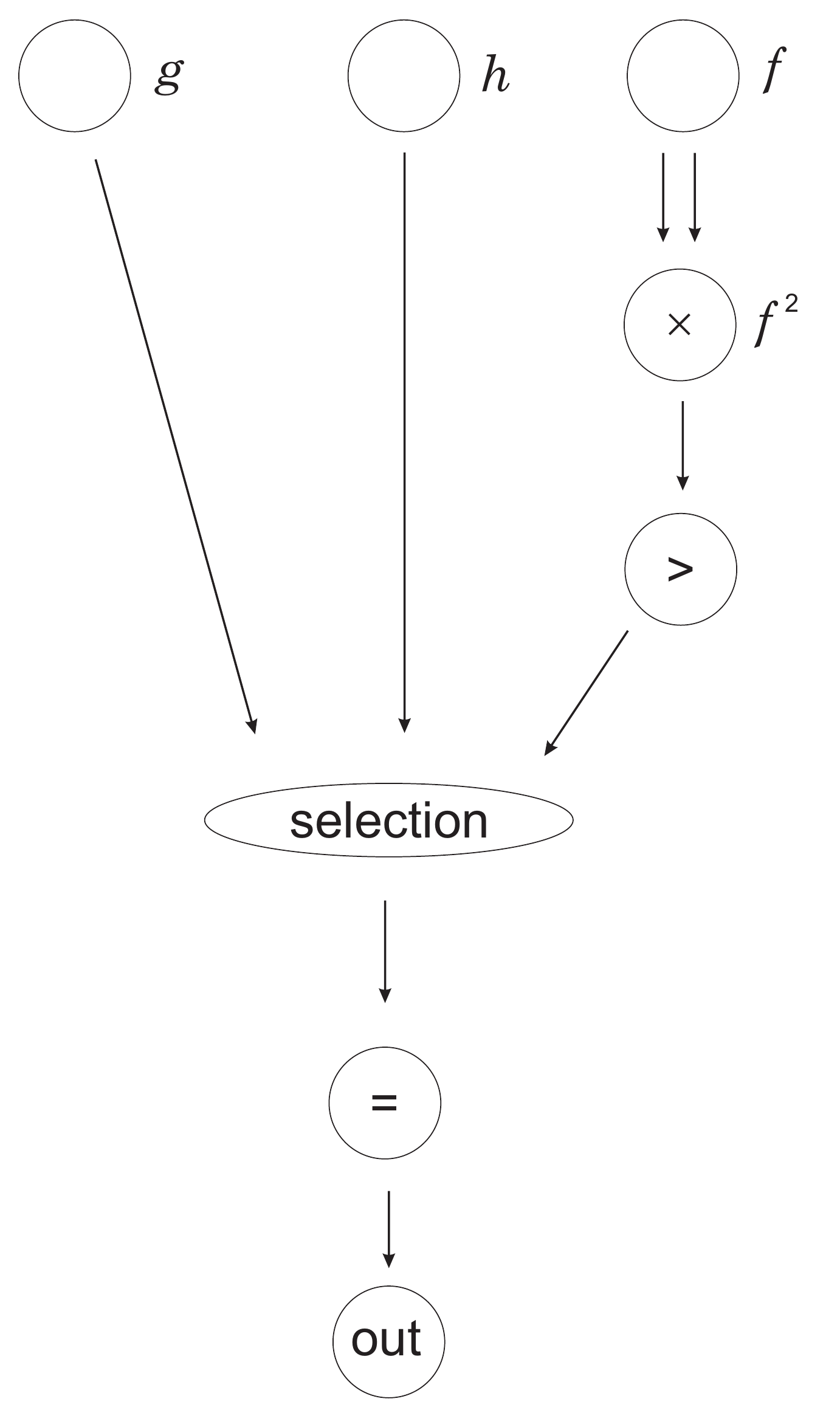}
             }
           }
\caption{ }
\label{fig:network1}
\end{figure}

\begin{figure}[hbt]
       \centerline{
          \scalebox{0.3}{
             \includegraphics{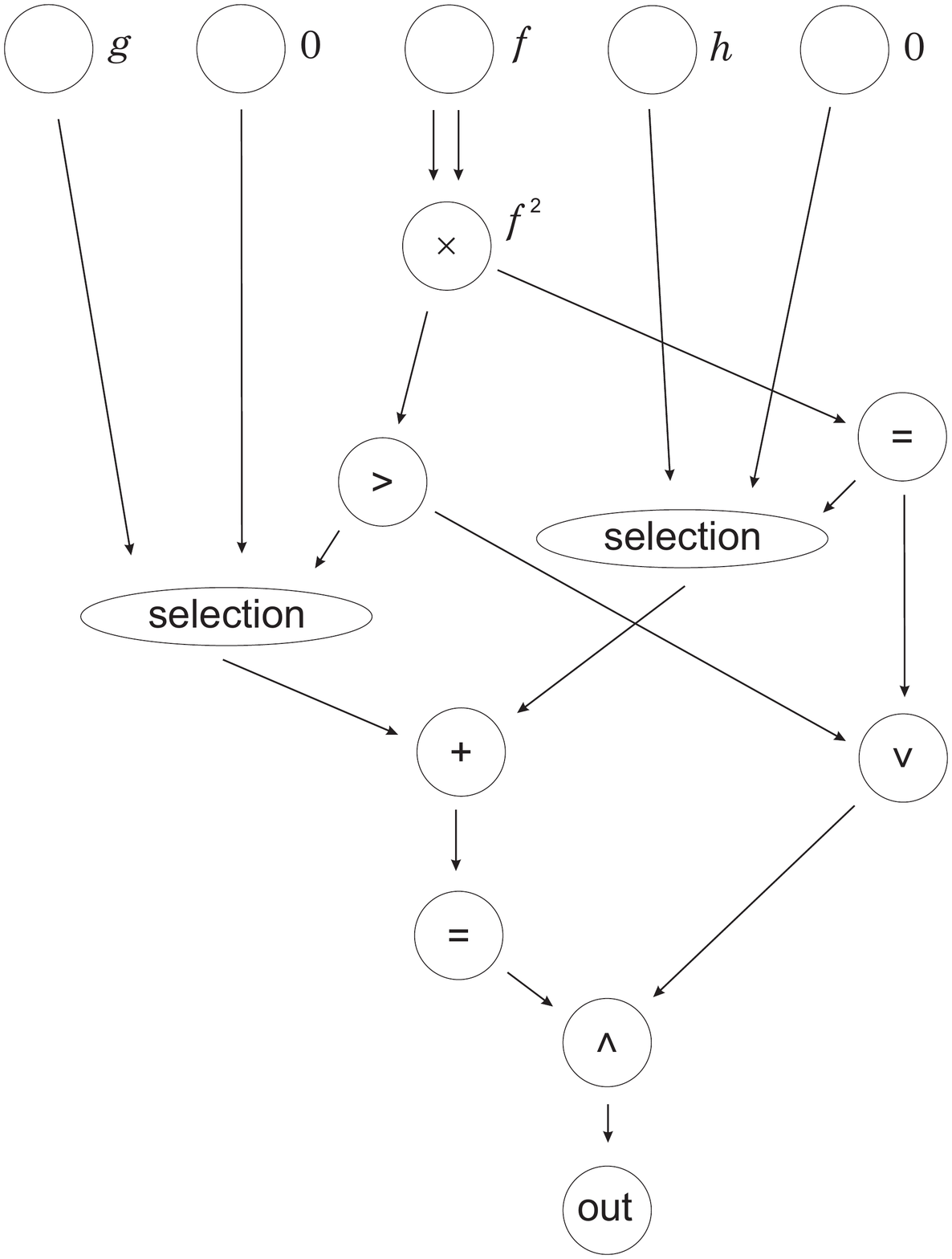}
             }
           }
\caption{ }
\label{fig:network1}
\end{figure}

\begin{lemma}\label{le:upper}
Let $\mathcal N$ be an arithmetic network testing membership in a semi-algebraic set $\Sigma \subset \Real^n$.
Then there exist an arithmetic network ${\mathcal T}$, testing membership in
$T_n(\Sigma)$, such that $d({\mathcal T}) \le c_1 d({\mathcal N})+ c_2 \log n$ for
some positive constants $c_1,\ c_2$.
\end{lemma}

\begin{proof}
Consider $n+1$ copies of ${\mathcal N}_{\delta, \eps}$ from Lemma~\ref{le:deltaeps}, sharing the same set of input gates.
In the $i$-th copy of ${\mathcal N}_{\delta, \eps}$ ($i=0, \ldots n$) replace all occurrences of $\eps$
(respectively, $\delta$) by $\eps_i$ (respectively, $\delta_i$) to obtain ${\mathcal N}_{\delta_i, \eps_i}$.
Note that this requires adding $2(n+1)$ new constant gates, $\eps_0, \delta_0, \ldots, \eps_n, \delta_n$.
Each ${\mathcal N}_{\delta_i, \eps_i}$ has a single output gate.
Collect these gates in one output using a binary tree of depth $O(\log n)$ with $\lor$-Boolean gates.
As a result we obtain the arithmetic network ${\mathcal T}$ accepting the set $T_n(\Sigma)$, and such that
$d({\mathcal T}) \le c_1 d({\mathcal N})+ c_2 \log n$ for some positive constants $c_1,\ c_2$.
\end{proof}

\begin{proof}[Proof of Theorem~\ref{th:general}]
Let $\mathcal T$ be the network constructed in Lemma~\ref{le:upper}.
By Proposition~\ref{pr:Mon2}, since $T_{n}(\Sigma)$ is compact, we have
$$d({\mathcal T}) = \Omega \left(\sqrt{ \frac{\log ({\rm b}(T_n(\Sigma)))}{n}}\right).$$
Therefore, by Proposition~\ref{pr:main},
$$d({\mathcal T}) = \Omega \left(\sqrt{ \frac{\log ({\rm b}(\Sigma))}{n}}\right).$$
By Lemma~\ref{le:upper}, we have $d({\mathcal T}) \le c'_1 d({\mathcal N})+ c'_2 \log n$ for some positive
constants $c'_1,\ c'_2$, hence
$$d(\mathcal N) \ge c_1 \sqrt{ \frac{\log ({\rm b}(\Sigma))}{n}} - c_2 \log n$$
for some positive constants $c_1,\ c_2$.
\end{proof}

\begin{theorem}\label{th:proj}
Let $\mathcal N$ be an arithmetic network testing membership in a semi-algebraic set $\Sigma \subset \Real^n$.
Let $\rho:\> \Real^n \to \Real^{n-r}$, for some $r=0, \ldots , n$, be the projection map.
Then
\begin{equation}\label{eq:th}
d({\mathcal N}) \ge  c_1  \frac{\sqrt{\log ({\rm b}(\rho(\Sigma)))}}{n} - c_2 \log n
\end{equation}
for some positive $c_1,\ c_2$.
\end{theorem}

Let
$$
W_p:=\underbrace{T_n(\Sigma) \times_{\rho(T_n(\Sigma))} \cdots \times_{\rho(T_n(\Sigma))}  T_n(\Sigma)
}_\text{(p+1) {\rm times}},
$$
where $p \le n$.

\begin{lemma}\label{le:proj}
Let ${\mathcal N}$ be an arithmetic network for $\Sigma$.
Then for each $p \le n$ there exists a network ${\mathcal N}^W$, testing membership in $W_p$, such that
$d({\mathcal N}^W) \le c_1d(\mathcal N)+ c_2 \log n$ for some positive $c_1,\ c_2$.
\end{lemma}

\begin{proof}
Lemma~\ref{le:upper} implies that there is an arithmetic network ${\mathcal T}$ testing membership in
$T_n(\Sigma)$ with $d({\mathcal T}) \le c_1d(\mathcal N))+ c_2 \log n$ for some positive $c_1,\ c_2$.
The problem of membership in $W_p$ has input variables
$$X_1, \ldots ,X_{n-r}, Y_{1, n-r+1}, \ldots , Y_{1, n}, \ldots , Y_{p, n-r+1}, \ldots , Y_{p, n}.$$

Construct the network ${\mathcal N}^W$ by taking $p$ copies of ${\mathcal T}$, so that the $i$th copy has
input variables $X_1, \ldots ,X_{n-r}, Y_{i, n-r+1}, \ldots , Y_{i, n}$.
Now using a dichotomy with depth $O(\log p)$ compute the conjunction of Boolean outputs
for all copies of ${\mathcal T}$.

Clearly the depth of the resulting network ${\mathcal N}^W$ is equal to the depth of ${\mathcal T}$ plus
the $O(\log p)$-depth needed to combine $p$ copies of ${\mathcal T}$ into one network, i.e.,
$c'_1 d({\mathcal T})+ c'_2 \log n$ for some positive $c'_1,\ c'_2$.
It follows that $d({\mathcal N}^W) \le c_1d(\mathcal N)+ c_2 \log n$ for some positive $c_1,\ c_2$.
\end{proof}

\begin{proof}[Proof of Theorem~\ref{th:proj}]
If $\log ({\rm b}(\rho(\Sigma))$, considered as a function of $n$, grows asymptotically slower than $n^2$, then
the right hand side of (\ref{eq:th}) is asymptotically negative for suitable positive $c_1,\ c_2$, and we are done.
Thus, assume this is not the case.

By Proposition~\ref{pr:map},
$$
{\rm b}_m(\rho(T_n(\Sigma)) \le \sum_{p+q=m} {\rm b}_q(W_p),
$$
it follows that
$${\rm b}(\rho(T_n(\Sigma))) \le \sum_{0 \le p \le n} {\rm b}(W_p).$$
Let ${\rm b}(W_\nu):= \max_p {\rm b}(W_p)$, then
\begin{equation}\label{eq:proj1}
\frac{{\rm b}(\rho(T_n(\Sigma)))}{n} \le {\rm b}(W_\nu).
\end{equation}

Since $W_\nu$ is compact, by (\ref{eq:Mon2}) we have
$$d({\mathcal N}^W) \ge c \sqrt{ \frac{\log ({\rm b}(W_\nu))}{n+ \nu r}}$$
for some positive constant $c$.
Replacing $d({\mathcal N}^W)$ in this inequality by a larger number
according to Lemma~\ref{le:proj}, we get for each $p \le n$:
$$d(\mathcal N) \ge  c'_1 \sqrt{ \frac{\log ({\rm b}(W_\nu))}{n+ \nu r}} - c'_2 \log n$$
for some positive $c'_1,\ c'_2$.
Then (\ref{eq:proj1}) implies that
$$d(\mathcal N) \ge  c'_1 \sqrt{ \frac{\log ({\rm b}(\rho(T_n(\Sigma)))- \log n}{n^2}} - c'_2 \log n.$$
According to (\ref{eq:rho}), $\rho(T_n(\Sigma))=T_n(\rho(\Sigma))$, hence
$$d(\mathcal N) \ge  c'_1 \sqrt{ \frac{\log ({\rm b}(T_n(\rho(\Sigma)))- \log n}{n^2}} - c'_2 \log n,$$
while, by Proposition~\ref{pr:main}, ${\rm b}(T_n(\rho(\Sigma))) \ge {\rm b}(\rho(\Sigma))$.
It follows that
\begin{equation}\label{eq:proj2}
d(\mathcal N) \ge  c'_1 \sqrt{ \frac{\log ({\rm b}(\rho(\Sigma))- \log n}{n^2}} - c'_2 \log n,
\end{equation}
Since we assumed that $\log ({\rm b}(\rho(\Sigma))$ grows faster than $n^2$, (\ref{eq:proj2}) implies
(\ref{eq:th}) for some positive $c_1,\ c_2$.
\end{proof}

\section{Application}

In this section we apply the bound from Theorem~\ref{th:general} to an example of a specific
computational problem (a particular case of ``Parity of Integers'' problem in \cite{GV14}).
\medskip

{\em For  given three real numbers $x_1, x_2, x_3$, where $1 \le x_i \le n,\ n \in \Z$, decide whether the following
property is true: either all $x_i$, or exactly one of $x_i$, are integer.}
\medskip

To obtain a lower bound, consider the integer lattice $\{1, \ldots, n\}^3$ in $\Real^3$ and let
$\Sigma$ be the union of all open two-dimensional squares and all vertices.
Then the problem is equivalent to deciding membership in $\Sigma$.
Obviously, the complement $\Real^3 \setminus \Sigma$ is connected.
Observe that $\Sigma$ is not locally closed.
It is homotopy equivalent to a two-plane with $\Omega ( n^3)$ punctured points,
so $b_1(\Sigma) = \Omega ( n^3)$.
By Theorem~\ref{th:general}, the depth of any arithmetic network testing membership in $\Sigma$
is $\Omega ( \sqrt{\log n})$ for some positive $c$.

Various simple algorithms provide an upper bound $O(\log n)$ for the problem.
For example, for each of $x_i$ in parallel or sequentially, decide whether or not it's one of the numbers $1, \ldots ,n$,
evaluating the disjunction
$$\bigvee_{1 \le j \le n} (x_i-j=0)$$
by means of dichotomy.
This requires the depth $O(\log n)$.
Let the result be $0$ if $x_i$ is integer and $1$ otherwise.
The network computes, with constant depth, the sum of the results.
Then the answer is Yes if and only if the sum is either 0 or 2.
The depth of the resulting network is $O(\log n)$.
\medskip

\subsection*{Acknowledgements}
We thank Dima Grigoriev, Joachim von zur Gathen and Luis Pardo for discussions of various aspects of
arithmetic networks.

\end{document}